\documentclass[submission,copyright,creativecommons]{eptcs}
\providecommand{\event}{QFM 2012}

\usepackage{etex,etoolbox}

\usepackage{mathrsfs}
\usepackage[cmex10]{amsmath}
\usepackage{amssymb}
\usepackage{amsfonts}
\usepackage{amssymb}
\usepackage{amsthm}
\usepackage{graphicx}
\usepackage{stmaryrd}
\usepackage{bussproofs}
\usepackage{txfonts}
\input{xypic.sty}
\xyoption{all}

\newcommand{\brck}[1]{ \llbracket{#1}\rrbracket}

\newcommand{\e}{\varepsilon}
\renewcommand{\epsilon}{\varepsilon}
\newcommand{\M}{\mathcal{M}}

\newcommand{\lng}{\mathcal{L}}

\newcommand{\ol}{\overline}

\newcommand{\nat}{\mathbb{N}}
\newcommand{\preals}{\mathbb{R}^+}
\newcommand{\prats}{\mathbb{Q}^+}

\newcommand*\dotminus{
  \buildrel\textstyle.\over{%
    \hbox{\vrule height3pt depth0pt width0pt}{\smash-}}}

\usepackage{tikz}
\usetikzlibrary{automata}
\usetikzlibrary{shapes.multipart}

\newtheorem{definition}{Definition} 
\newtheorem{theorem}{Theorem} 
\newtheorem{lemma}{Lemma}
\newtheorem{corollary}{Corollary}
\newtheorem{example}{Example}

\providecommand{\urlalt}[2]{\href{#1}{#2}}
\providecommand{\doi}[1]{doi:\urlalt{http://dx.doi.org/#1}{#1}}

\begin{document}

\title{Parameterized Metatheory for Continuous Markovian Logic}



\author{Kim G. Larsen \quad Radu Mardare \quad Claus Thrane
\institute{Department of Computer Science, Aalborg University}
\institute{Selma Lagerl{\"o}fs Vej 300 - 9220 Aalborg, Denmar}
\email{\{kgl,mardare,crt\}@cs.aau.dk}
}
\def\titlerunning{Parameterized Metatheory for Continuous Markovian Logic}
\def\authorrunning{Larsen, Mardare \& Thrane}

\maketitle

\begin{abstract}
  This paper shows that a classic metalogical framework, including all
  Boolean operators, can be used to support the development of a
  metric behavioural theory for Markov processes. Previously, only
  intuitionistic frameworks or frameworks without negation and logical
  implication have been developed to fulfill this task. The focus of
  this paper is on continuous Markovian logic (CML), a logic that
  characterizes stochastic bisimulation of Markov processes with
  an arbitrary measurable state space and continuous-time
  transitions. For a parameter $\e>0$ interpreted as observational
  error, we introduce an $\e$-parameterized metatheory for CML: we
  define the concepts of $\e$-satisfiability and $\e$-provability
  related by a sound and complete axiomatization and prove a series of
  ``parameterized'' metatheorems including decidability, weak
  completeness and finite model property. We also prove results
  regarding the relations between metalogical concepts defined for
  different parameters. Using this framework, we can characterize both
  the stochastic bisimulation relation and various observational
  preorders based on behavioural pseudometrics. The main contribution
  of this paper is proving that all these analyses can actually be
  done using a unified complete Boolean framework. This extends the
  state of the art in this field, since the related works only propose
  intuitionistic contexts that limit, for instance, the use of the
  Boolean \emph{logical implication}.
\end{abstract}

\section{Introduction}
Stochastic models have successfully been used to describe the
qualitative and quantitative behavior of systems in many natural and
artificial domains. The problems addressed in this paper refer to the
most general models of Markov processes, defined for arbitrary
(analytic) state-spaces and continuous-time transitions, henceforth
\emph{continuous Markov processes} (CMPs); they subsume well-known
models such as \emph{continuous-time Markov chains} and \emph{labelled
  Markov processes} and for this reason our work can be simply
instantiated for these particular models. CMPs have been initially
introduced by Desharnais and Panangaden in \cite{Desharnais03b}. In
this paper, for technical reasons, we use the definition of CMPs
proposed by the first two authors and Cardelli in \cite{Cardelli11a},
which exploits an equivalence between the definitions of
\emph{Harsanyi type spaces} \cite{Moss04} and a coalgebraic view of
labelled Markov processes~\cite{deVink99} proved, for instance, by
Doberkat in \cite{Doberkat07}.

In this paper the class of CMPs define the semantics for Continuous
Markovian Logic (CML) \cite{Cardelli11a,Cardelli11b}. This is a
multimodal logic endowed with modalities $L_r, M_r$ for $r\in\prats$,
similar to the ones used by the Aumann system \cite{Aumann99a}, that
approximates the transition rates. For instance, a process satisfies
$L_r\phi$ if the rate of its transition from the initial state to a
state satisfying $\phi$ is at least $r$. In \cite{Cardelli11a} it has
been proved that this logic characterizes the stochastic bisimulation
of CMPs.

Despite the elegant theories supporting the concepts of stochastic and
probabilistic bisimulations and their relation to logics
\cite{Larsen91}, these concepts remain too strict for applications. In
modelling, the values of the rates or probabilities are often
approximated and consequently, one is interested to know whether two
processes that differ by a small amount in real-valued parameters show
similar (not necessarily identical) behaviours. In such cases, instead
of a bisimulation relation, one needs a metric concept to estimate the
degree of similarity of two systems in terms of their behaviours. The
metric theory for Markov processes was initiated by Desharnais et
al.~\cite{Desharnais04} and has been greatly developed and explored by
van Breugel, Worrell and
others~\cite{vanBreugel01b,vanBreugel03}. Similar notions have been
proposed in literature for less general models. Bringing with them
notions such as the point-wise simulation distance defined in
\cite{Jou90} for discrete probabilistic systems, and the discounted
distances proposed in \cite{Alfaro03,Alfaro09} for weighted systems,
and of the quantified similarities of timed systems studied in
\cite{Henzinger05} and \cite{TFL10} (see also \cite{crt} for an overview).

One way of defining these behavioural distances was proposed by
Kozen~\cite{mu} and consists in replacing the classic logical
framework used to encode properties of processes with a non-classical
real-valued framework that will interpret logical formulae as
functional expressions mapping states to reals. In this way, we get a
relaxation of the satisfiability relation which is replaced by a
function that reports the ``degree of satisfiability'' between a Markov
process and a logical property. This further induces a
\emph{behavioural pseudometric} on processes with the stochastic
bisimulation as its kernel and measuring the distance between
processes in terms of their behavioural similarity. Such formalisms
have since been proposed for Markov systems by Desharnais, Panangaden
and others \cite{Desharnais04,Panangaden09}.

It was hoped that these metrics would provide a quantitative
alternative to logic, but this did not happen. One reason could
originate in the fact that all this ``metric reasoning'' focused
exclusively on the semantics of the logic while a syntactic or a
metalogical counterpart did not develop until recently. Such a logical
perspective on distance was proposed by the first two authors of this
paper and Cardelli in \cite{Cardelli11a}, where it was emphasized
that, in the context of a completely axiomatized logic, the semantic
distance between Markovian processes implicitly induces, via Hausdorff
metrics, a distance between logical properties that can be interpreted
as a \emph{measure of provability} for CML. On this line,
\cite{Cardelli11a} and \cite{Cardelli11b} contain the open ideas of a
research program that we have followed ever since. This research aims
to understand the relation between the pseudometric space of Markov
processes and the pseudometric space of logical formulae i.e., the
relation between the measure of similarity for Markov processes and
the measure of provability in a corresponding stochastic/probabilistic
logic. Eventually, in \cite{Larsen12a}, the first two authors in
collaboration with Prakash Panangaden have identified a metric analog
of Stone duality that relates the two pseudometric spaces and in
\cite{Larsen12b} we have studied how convergence in the open ball
topologies induced by the two pseudometrics ``agree to the limit''.

However, we have yet to clarify what the kernel of the distance
between logical formulae is. It was shown in \cite{Larsen12b} that it
is possible to have formulae at distance $0$ that are not logically
equivalent, and we have characterized this kernel for some limited
fragments of CML. But the full picture has not been yet achieved. One
reason for this difficulty originates from the fact that we have no
pure metalogical definition of this distance, as it is always
implicitly obtained from the definition of the behavioural
pseudometrics, hence it depends of the semantics. This is exactly what
we achieve in this paper: a metalogical definition of the behavioural
distance.

This paper is a step forward in the process of understanding this
distance from a logical perspective. We define a parameterized
metatheory for CML, where the parameter $\e\geq 0$ is interpreted as
an observational error which allows us to express properties
approximating the behavior of a given CMP. This metatheory consists in
defining an \emph{$\e$-semantics}, i.e., \emph{$\e$-satisfiability
  relation} denoted $\models_\e$, and to develop a complete
Hilbert-style axiomatization for a corresponding concept of
\emph{$\e$-provability}, denoted $\vdash_\e$. The classic semantics
and provability relation of CML are, in this context, the
$0$-semantics and $0$-proof theory.

This parametric metatheory allows us to transfer logical properties
between various semantics defined for different parameters. For
instance, we can translate $0$-satisfiability into $\e$-satisfiability
and reverse, or $0$-provability into $\e$-provability and reverse
using appropriate encodings. Exactly this allows one to see the
behavioural pseudometrics from a logical perspective. We show that
the distance between two CMPs $m_1$ and $m_2$ can, in fact, be defined
by the infimum of the set of values $\e$ such that for any CML formula
$\phi$, $m_i\models\phi$ iff $m_j\models_\e\phi$, where
$\{i,j\}=\{1,2\}$.

We develop the parametric metatheory as a classic metatheory and in
addition to the sound-complete axiomatization we prove a series of
metatheorems including an $\e$-deduction theorem, a $\e$-finite model
property and some $\e$-decidability results for CML.

The major contribution of this paper consists in the fact that this
entire development respects the classic Boolean restrictions. So fare all 
attempts of realizing Kozen's idea \cite{mu} and defining a
quantitative version of the satisfiability relation, have faced
noticeable problems related to the treatment of negation. Often in the
papers treating this argument, negation is either eliminated,
restricted to atomic propositions or considered in a non-Boolean
context \cite{FGK10, FLT10, Desharnais08}. This restriction is an
impediment for the use of classic reasoning. For instance in
\cite{Desharnais08}, the definition of $\e$-satisfiability contains
the following rules defined for an arbitrary Markov process $m$:
\newcommand{\sem}[1]{\llbracket #1 \rrbracket}
\begin{align*}
  m \models_\e &\lnot \phi ~\text{iff}~ m\not\models_{-\e} \phi\\
  m \models_\e & \langle a \rangle_{\delta}\phi ~\text{iff}~ \theta(m,\sem{\phi}_\e) > \delta - \e
\end{align*}
where $\theta(m,M)$ is the probability of a transition from $m$ to a
state in the set $M$. Observe that the rule for negation requires to
transport information from ``$\e$-semantics'' to ``$-\e$-semantics'' and
that it is obviously non-Boolean. For instance, in the case
$\theta(m,\sem{\phi}_\e)=\delta$ one can prove, using the previous
rules, that for $\e>0$ we have
$$m\models_\e\langle a \rangle_\delta \phi \land \lnot \langle a \rangle_\delta \phi.$$ 
In other words, the logic is inconsistent if it is interpreted in a
Boolean context.

In the light of this observation, one can see the real contribution of
our paper. We show that it is possible to obtain the sought after
behavioral distances and remain Boolean and classic to all logical
levels. Of course, one can argue that an intuitionistic approach is as
good for applications as the classic Boolean approach is, and we
cannot argue against this. But we believe that for a deeper
understanding of Markov processes and for providing a strong
theoretical background for an approximation theory of Markov
processes, a classic logical framework is more useful. In fact, in
\cite{Larsen12a} a special Boolean algebra (called \emph{Aumann
  algebra}) is identified with operators that corresponds to CML and
we proved Stone duality results between these algebras and Markov
processes. This enforces our trust that the Boolean setting is the
right one for studying properties of Markov processes.

\bigskip

To summarize, the achievements of this work are as follows.
\begin{itemize}
\item We develop a parameterized metatheory for continuous Markovian logic that extends the classic metatheory. The parameter can be interpreted as observational error.
\item We define the concept of $\e$-satisfiability and identify for it an appropriate concept of $\e$-provability with a sound and complete Hilbert-style axiomatization.
\item We prove that classic metatheorems about CML remain true in the parametric semantics. Such properties are the weak completeness, the finite model property and decidability.
\item We show that this parameterized metatheory can be used to define a behavioural pseudometric, which is a distance between CMPs that characterizes the similarity of two processes from the point of view of their behaviours.
\item We identify two behavioural orders that have, in the parametric semantics, similar logical interpretations to the bisimulation in the classic semantics.
\item This entire development is essentially Boolean.  
\end{itemize}


\section{Preliminary definitions}\label{preliminaries} 

In this section we introduce some basic notations and concepts used
throughout this paper.

For arbitrary sets $M,N$, we denote by $2^M$ the powerset of $M$, by
$M\uplus N$ their disjoint union and by $[M\to N]$ the set of
functions from $M$ to $N$.

Given a relation $\mathfrak R\subseteq M\times M$, the $\mathfrak
R$-closure of a set $N\subseteq M$ is the set $N^{\mathfrak R}=\{m\in
M~|~\exists n\in N, (n,m)\in\mathfrak R\}$; we say that $N$ is
$\mathfrak R$-closed iff $N^{\mathfrak R}\subseteq N$. If
$\Sigma\subseteq 2^M$, then $\Sigma(\mathfrak R)$ denotes the set of
$\mathfrak R$-closed elements of $\Sigma$.

A set $\Sigma\subseteq 2^M$ is a \emph{$\sigma$-algebra over $M$} if
it contains $M$ and it is closed under complement and countable union.
Given a $\sigma$-algebra $\Sigma$ over $M$, the tuple $(M,\Sigma)$ is
called a \emph{measurable space} and the elements of $\Sigma$,
\emph{measurable sets}. A set $\Omega\subseteq 2^M$ is a
\emph{generator for $\Sigma$} if $\Sigma$ is the closure of $\Omega$
under complement and countable union.

Given a measurable space $(M,\Sigma)$, a function
$\mu:\Sigma\to\preals$ is a measure iff $\mu(\emptyset)=0$ and for any
sequence $\{N_i\mid i\in I\subseteq\nat\}$ of pairwise disjoint
measurable sets, $\mu(\bigcup_{i\in I} N_i)=\sum_{i\in I}{\mu(N_i)}.$
The set of all measures on $(M,\Sigma)$ is denoted by
$\Delta(M,\Sigma)$.  We organize $\Delta(M,\Sigma)$ as a measurable
space by considering the $\sigma$-algebra $\mathfrak F$ generated, for
arbitrary $S\in\Sigma$ and $r>0$, by the sets
$F^r_S=\{\mu\in\Delta(M,\Sigma):\mu(S)\geq r\}.$

Given two measurable spaces $(M,\Sigma)$ and $(N,\Sigma')$, a mapping
$f:M\to N$ is \emph{measurable} if $\mbox{for any }T\in\Sigma',
f^{-1}(T)\in\Sigma$. We use $\llbracket M\to N\rrbracket$ to denote
the class of measurable mappings from $(M,\Sigma)$ to $(N,\Sigma')$,
assuming of course that $\Sigma$ and $\Sigma'$ are clear from the
context.

Central for this paper is the notion of \emph{analytic space} that
supports some of the main results. As properties of analytic spaces
are however not used here directly, we only recall the main
definitions. For detailed discussion on this topic related to Markov
processes, the reader is referred to \cite{Panangaden09} (Section 7.5)
or to \cite{Doberkat07} (Section 4.4).

A metric space $(M,d)$ is \emph{complete} if every Cauchy sequence
converges in $M$. A \emph{Polish space} is the topological space
underlying a complete metric space with a countable dense subset. An
\emph{analytic space} is the image of a Polish space under a
continuous function between Polish spaces.


\section{Continuous Markov processes}\label{CMP}

In this section we introduce continuous Markov processes (CMPs) \cite{Cardelli11a,Cardelli11b} which are models of stochastic systems with analytic state space and continuous-time transitions. The definition is similar to the one proposed by Desharnais and Panangaden in \cite{Desharnais03b}, but it exploits an equivalence between the definitions of Harsanyi type spaces
\cite{Moss04} and a coalgebraic view of labelled Markov processes~\cite{deVink99} proved, for instance, by Doberkat in \cite{Doberkat07}. However, with respect to \cite{Cardelli11a,Cardelli11b} or to \cite{Panangaden09,Desharnais02,Doberkat07}, we do not consider action labels. The labels can easily be added without changing any aspects of the theory.

\begin{definition}[Continuous Markov processes]
  Given an analytic set $(M,\Sigma)$, where $\Sigma$ is the Borel
  $\sigma$-algebra generated by the topology, a \emph{continuous
    Markov kernel} (CMK) is a tuple $\M=(M,\Sigma,\theta)$, where
  $\theta\in\llbracket M\to\Delta(M,\Sigma)\rrbracket$ is the
  \emph{transition function}. The set $M$ is the \emph{support-set of
    $\M$} denoted $supp(\M)$.
Whenever $m\in M$, then $(\M,m)$ is a \emph{continuous Markov process}.
\end{definition}

Notice that $\theta$ is a measurable mapping between $(M,\Sigma)$ and $(\Delta(M,\Sigma),\mathfrak F)$, where $\mathfrak F$ is the sigma algebra on $\Delta(M,\Sigma)$ defined in the preliminaries. This condition on $\theta$ is equivalent with the conditions on the two-variable \emph{rate function} used in \cite{Panangaden09, Desharnais02, Desharnais03b} to define
continuous Markov processes.



\subsection{Bisimulation}

Stochastic bisimulation for CMPs follows the line of Larsen-Skou
probabilistic bisimulation
~\cite{Larsen91,Desharnais02,Panangaden09}. Recall that
$\Sigma(\mathfrak R)$ in the next definition denotes the $\mathfrak
R$-closed sets of $\Sigma$.

\begin{definition}[Stochastic Bisimulation]
  Given a CMK $\M=(M,\Sigma,\theta)$ a binary relation $\mathfrak
  R\subseteq M\times M$ is a \emph{stochastic bisimulation relation}
  if whenever $(m,n)\in\mathfrak R$, for any $C\in\Sigma(\mathfrak
  R)$, $$\theta(m)(C)=\theta(n)(C).$$ Two processes $(\M,m)$ and
  $(\M,n)$ are \emph{stochastic bisimilar}, written $m\sim_\M n$, if
  they are related by a stochastic bisimulation relation.
\end{definition}

Observe that, for any CMK $\M$ there exist stochastic bisimulation
relations: for instance, the identity relation on its support-set is
such a relation. The relation $\sim_\M$ is the largest stochastic
bisimulation relation.
 
\begin{definition}[Disjoint union]
If $\M=(M,\Sigma,\theta)$ and $\M'=(M',\Sigma',\theta')$ are CMKs, then $\M''=(M'',\Sigma'',\theta'')$ defined by $M''=M\uplus M'$, $\Sigma''$ is the $\sigma$-algebra generated by $\Sigma\uplus\Sigma'$ and 
\begin{align*}
  \theta''(m)(N\uplus N')=
  \begin{cases}
    \theta(m)(N) & \textrm{if } m\in M\\
    \theta'(m)(N') & \textrm{if }m\in M'
  \end{cases}
\end{align*}
for arbitrary $N\in\Sigma$ and $N'\in\Sigma'$, is the disjoint union of $\M$ and $\M'$ denoted by $\M''=\M\uplus\M'$.
\end{definition}
Observe that the disjoint union of CMKs is a CMK. The previous
definition allows us to define stochastic bisimulation between
processes from different CMKs. If $m\in M$ and $m'\in M'$, we say that
$(\M,m)$ and $(\M',m')$ are \emph{bisimilar} written $(\M,m)\sim
(\M',m')$ whenever
$m\sim_{\M\uplus\M'}m'$. 

\subsection{Generators}
The definition of bisimulation can be amended to focus on two particular classes of generators of the $\sigma$-algebra.

\begin{definition}[Bisimulation Generators]
  Consider the CMK $\M=(M,\Sigma,\theta)$ and let $$\Theta=\{\theta(m)^{-1}([0,r])\mid m\in M, r\in\preals\}.$$ 
  \begin{itemize}
  \item The \emph{bisimulation generator of $\M$}, denoted by $G_\M$, is the closure of $\Theta$ under union and intersection. 
  \item The \emph{extended bisimulation generator of $\M$}, denoted $\ol{G_\M}$, is the closure of $\Theta$ under union, intersection and complement.
  \end{itemize}
\end{definition}

Observe that the bisimulation generators are generators of $\Sigma(\sim)$ which is a sub-sigma algebra of $\Sigma$, i.e., $\Sigma(\sim)$ is the closure of both $G_\M$ and $\ol{G_\M}$ under
complement and countable union. This observation allows us to characterize the stochastic bisimulation from the perspective of the bisimulation generators and to propose some generalizations of the stochastic bisimulation. But more importantly, they will be instrumental later, for obtaining our results on logical characterization.

\begin{theorem}\label{t1}
  Given a CMK $\M=(M,\Sigma,\theta)$, a relation $\mathfrak R\subseteq M\times M$ is a stochastic bisimilarity relation iff one (or both) of the following equivalent conditions is satisfied.
  \begin{itemize}
  \item whenever $(m,n)\in\mathfrak R$, $\theta(m)(C)=\theta(n)(C)$ for any $C\in G_\M$,
  \item whenever $(m,n)\in\mathfrak R$, $\theta(m)(C)=\theta(n)(C)$ for any $C\in \ol{G_\M}$.
  \end{itemize}
\end{theorem}





\section{Continuous Markovian Logics}\label{logic}

In this section we recall the continuous Markovian logic (CML)
introduced and studied in \cite{Cardelli11a,Cardelli11b}. This logic
extends the probabilistic logics for discrete-time Markov processes
\cite{Larsen91,Desharnais02,Panangaden09} and for Harsanyi type spaces
\cite{Fagin94,Zhou07} to stochastic domains and it characterizes the
stochastic bisimulation. In addition to the Boolean operators, this
logic is endowed with \emph{stochastic modal operators} that
approximate the rates of transitions. In the original definition,
$L_r\phi$ is a property of a CMP $(\M,m)$ whenever the rate of the
transition from $m$ to the class of states satisfying $\phi$ is
\emph{at least $r$}.


\begin{definition}[Syntax]
  The set $\lng$ of formulae of CML is generated by the following
  grammar, for arbitrary $r\in\prats$.
  $$\lng:~~~\phi:=\top~|~\lnot\phi~|~\phi\land\phi~|~L_r\phi.$$
\end{definition}

As usual, we work with all the Boolean operators, including
$\bot=\lnot\top$. In addition, we isolate two useful sublanguages of
$\lng$:
$$\lng^+:~~\psi:=\top~|~\psi\land\psi~|~\psi\lor\psi~|~L_r\psi \quad~\text{and}\quad \lng^-=\{\lnot\phi\mid\phi\in\lng^+\}.$$

\subsection{Parameterized Semantics: $\e$-satisfiability}

In\cite{Cardelli11a,Cardelli11b} the first two authors in
collaboration with Cardelli defined the semantics of CML for arbitrary
CMPs, henceforth the \emph{classic semantics} for CML. We will take a
similar approach in this paper with the difference that the
satisfiability relation is parameterized. Thus, for each rational
$\e\geq 0$, we introduce an $\e$-semantics that provides an
approximation of the classic semantics. The $\e$-semantics can be seen
as an ``approximation from below'' of the classic semantics: while
$L_r\phi$ is interpreted at $m$ as \emph{``the rate of the transitions
  from $m$ to the class of the states satisfying $\phi$ is at least
  $r$''}, in the $\e$-semantics it means that \emph{``the rate of the
  transitions from $m$ to the class of the states satisfying $\phi$ is
  at least $r-\e$''}. In this way one can encode observational errors
in the logic. Unlike the similar approach of \cite{Desharnais08}, we
propose a Boolean semantics.

\begin{definition}[$\e$-Satisfiability]
  For an arbitrary rational $\e\geq 0$, the \emph{$\e$-satisfiability relation} $\models_\e\subseteq\mathfrak M\times\lng$ is defined inductively on the structure of $\phi\in\lng$, as follows.
  \begin{itemize}
  \item $m\models_\e \top$ always,
  \item $m\models_\e\lnot\phi$ iff it is not the case that
    $m\models_\e\phi$, 
  \item $m\models_\e\phi\land\psi$ iff $m\models_\e\phi$ and
    $m\models_e\psi$,
  \item $m\models_\e L_r\phi$ iff $\theta(m)(\llbracket
    \phi\rrbracket_\e)+\e\geq r$,
  \end{itemize}
where $\llbracket \phi\rrbracket_\e=\{m\in\mathfrak M~|~m\models_\e\phi\}$.
\end{definition}

Notice that the classic semantics for CML introduced in
\cite{Cardelli11a,Cardelli11b} is nothing else but $0$-semantics,
since $\models_0{}={}\models$.

\begin{example}
  Consider the CMK $\M=(M,2^M,\theta)$ represented in Figure
  \ref{ex1}, where $M=\{m,m_1,m_2,m_3,m_4,m_5\}$ and $\theta$ is
  defined by the values $r,s,s',u\in\prats$ that label the transition
  arrows\footnote{For simplicity we only represented the transitions
    with strict positive rates.}. We can now understand the difference
  between the classic and the $\e$-semantics. For instance,
$$m\models L_{s+s'} L_u\top$$ 
since $m_2\sim m_4$, $\theta(m)(\{m_2,m_4\})=s+s'$ and $m_2\models
L_u\top$ because $m_3\sim m_5$ and $\theta(m_2)(\{m_3,m_5\})=u$.

Similarly, for some $\e>0$, $$m\models_\e L_{s+s'+\e} L_{u+\e}\top$$
since $\theta(m)(\{m_2,m_4\})=s+s'\geq(s+s'+\e)-\e$ and $m_2\models_\e
L_{u+\e}\top$ because $\theta(m_2)(\{m_3,m_5\})=u\geq (u+\e)-\e$.

On the other hand, $$m\not\models L_{s+s'+\e} L_{u+\e}\top$$ since
$\theta(m)(\{m_2,m_4\})=s+s'\not\geq(s+s'+\e)$ and $m_2\not\models
L_{u+\e}\top$ because $\theta(m_2)(\{m_3,m_5\})=u\not\geq (u+\e)$.
Which is exactly the way one may see $\e$ as an observational
error.
\end{example}

\begin{figure}[t]
  \centering
      \begin{tikzpicture}[->,shorten >=1pt, auto, node distance=.75cm,
        initial text=, scale=0.75]
        
        \begin{scope}[shape=circle, outer sep=1pt,minimum
          size=5pt,inner sep=1pt, node distance=1.9cm] 
          \node (t0) {$m$};        
          \node (t1) at (200:1.9){$m_1$};
          \node (t5) [below of=t0] {$m_2$};
          \node (t4) [below of=t5] {$m_3$};
          \node (t6) at (340:1.9){$m_4$};
          \node (t7) [below of=t6]{$m_5$};
        \end{scope}        
	            
        \begin{scope}
          \draw (t0) -- node[above] {$r$} (t1);
          \draw (t0) -- node[above] {$s$} (t6);
          \draw (t0) -- node[right] {$s'$} (t5);
          \draw (t5) -- node[right] {$u$} (t4);
          \draw (t6) -- node[right] {$u$} (t7);
       \end{scope}

     \end{tikzpicture}
\caption{A Markov process}\label{ex1}
\end{figure}
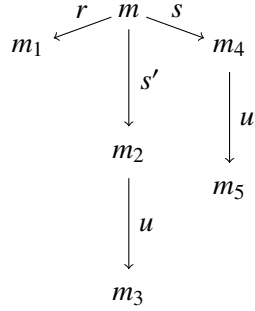

The semantics of $L_r\phi$ is well defined only if $\llbracket
\phi\rrbracket_\e$ is measurable. This is guaranteed by the
fact that $\theta$ is a measurable mapping between $(M,\Sigma)$ and
$(\Delta(M,\Sigma),\mathfrak F)$, as proved in the next lemma.

\begin{lemma}\label{measurability}
For any $\phi\in\lng$, $\llbracket\phi\rrbracket_\e\in\Sigma$.
\end{lemma}


  



We extend the classical metalogical concepts to the parametric metatheory.

\begin{definition}\label{meta1}
Given a rational $\e\geq 0$, a formula $\phi$ is \emph{$\e$-satisfiable} if there exists
$m\in supp(\M)$ such that $m\models_\e\phi$. We say that $\phi$ is
\emph{$\e$-valid}, denoted by $\models_\e\phi$, if $\lnot\phi$ is not
$\e$-satisfiable. 
\end{definition}

For notational convenience we will write $m\not\models_\e\phi$ when it
is not the case that $m\models_\e\phi$, and use $\models$ in place of
$\models_0$. The proof of the previous lemma reveals a deeper result
connecting the $\e$-semantics, if we involve our notion of
bisimulation generators.

\begin{corollary}\label{c1}
  For any rational $\e\geq0$,
  $G_\M=\{\llbracket\phi\rrbracket_\e\mid\phi\in\lng^+\}$ and
  $\ol{G_\M}=\{\llbracket\phi\rrbracket_\e\mid\phi\in\lng\}$.
\end{corollary}



The major advantage that the parametric semantics provides is that one
can handle in parallel properties from different semantics and, for
instance, can prove ($\e+\e'$)-satisfiability properties from
properties concerning $\e$-satisfiability.

In what follows we establish a few such results. The first lemma establishes the relation between $\e$-semantics and the classic semantics.

\begin{lemma}\label{l1}
  If $\phi\in\lng^+$, then for arbitrary $\e,\e'\geq 0$,
  $m\models_\e\phi$ implies $m\models_{\e+\e'}\phi$. In particular,
  $m\models\phi$ implies $m\models_\e\phi$.
\end{lemma}


The following counter-example shows that we cannot hope for the result to hold for negative formulae.

\begin{example}
  Consider the CMK $\M=(\{m\},2^{\{m\}},\theta)$ with
  $\theta(m)(\llbracket \top \rrbracket) =r$. Clearly $m \models \lnot
  L_{r+\delta} \top$ for all $\delta >0$. Suppose that we also have $m
  \models_\epsilon \lnot L_{r+\delta} \top$. This is equivalent to $m
  \models \lnot L_{r+\delta-\epsilon} \top$, i.e.,
  $\theta(m)(\llbracket \top \rrbracket) < r+\delta - \epsilon$ for
  all $\delta >0$. This last inequality implies $\theta(m)(\llbracket
  \top \rrbracket) \leq r-\epsilon$ which contradicts our initial
  assumption.
\end{example}

Notice that if $\e$ is growing, the set $\brck\phi_\e$ is increasing when $\phi\in\lng^+$ and is decreasing when $\phi\in\lng^-$.

Although negation turns out to be problematic in the case of the
previous lemma, we can however characterize the relation between
$\models_\e$ and $\models_{\e+\e'}$ for the entire language. To
characterize completely the relation between two parametric semantics,
we define a pair of dual encodings.

\begin{definition}
Let $\langle~\rangle_\e:\lng\to\lng$ and $\langle~\rangle^\e:\lng\to\lng$ be two functions on $\lng$ defined as follows.
  \begin{displaymath}
    \begin{array}{lll}
      \langle\top\rangle_\e=\top &~~~~& \langle\top\rangle^\e=\top \\
      \langle\phi\land\psi\rangle_\e=\langle\phi\rangle_\e\land\langle\psi\rangle_\e && \langle\phi\land\psi\rangle^\e=\langle\phi\rangle^\e\land\langle\psi\rangle^\e\\
      \langle\lnot\phi\rangle_\e=\lnot\langle\phi\rangle_\e && \langle\lnot\phi\rangle^\e=\lnot\langle\phi\rangle^\e\\
      \langle L_r\phi\rangle_\e=L_{r \dotminus \e}\langle\phi\rangle_\e && \langle L_r\phi\rangle^\e=L_{r+\e}\langle\phi\rangle^\e
    \end{array}
  \end{displaymath}
  where $r\dotminus\e= \max\{0,r-\e\}$ 
\end{definition}

Observe that for any $\phi$ that is not of type $L_r\psi$ with $r<\e$,
we have that
$\langle\langle\phi\rangle^\e\rangle_\e=\langle\langle\phi\rangle_\e\rangle^\e=\phi$.

Before we turn to the main theorem of this section, we apply the
previous definition to obtain a result on limits. The result may
additionally be considered a form of reverse implication for Lemma
\ref{l1}.

\begin{lemma}\label{l2}
  If $\phi\in\lng^+$ and for every rational $\e>0$,
  $m\models_{\e'+\e}\phi$, then $m\models_{\e'}\phi$. In particular,
  if $m\models_\e\phi$ for all rationals $\e > 0$, then also
  $m\models\phi$.
\end{lemma}


We are now ready to state the main theorem of this section that establishes the relation between various parameterized semantics for the entire language $\lng$.

\begin{theorem}\label{t2}
For arbitrary $\phi\in\lng$,
\begin{enumerate}
\item $m\models_{\e+\e'}\phi$ iff $m\models_\e\langle\phi\rangle_{\e'}$,
\item $m\models_\e\phi$ iff $m\models_{\e+\e'}\langle\phi\rangle^{\e'}$.
\end{enumerate}
\end{theorem}



From this last theorem we can derive a characterization of the relation between the classic semantics and the $\e$-semantics.

\begin{corollary}\label{c2}
For arbitrary $\phi\in\lng$,
\begin{enumerate}
\item $m\models_\e\phi$ iff $m\models\langle\phi\rangle_{\e}$,
\item $m\models\phi$ iff $m\models_{\e}\langle\phi\rangle^{\e}$.
\end{enumerate}
\end{corollary}


\section{Parameterized Proof Theory: $\e$-Provability}

In this section we extend the metatheory and define a
\emph{parameterized proof system} for our logic that corresponds to the
parameterized semantics.  The parameterized proof system will permit us
to prove, syntactically, approximated properties of models. We should
emphasize that we will not work with ``approximated proofs'', but with
``exact proofs'' about ``approximated properties'' and this is where
the Boolean character of our metatheory plays its role.

For each rational $\e\geq 0$ we introduce a notion of $\e$-provability
denoted by $\vdash_\e$. Table \ref{AS} contains a Hilbert-style
axiomatization of $\e$-provability for our $\e$-semantics. The axioms
and rules, which are considered in addition to the axiomatization of
classic propositional logic, are stated for propositional
variables $\phi,\psi\in\lng$ and arbitrary $s,r\in\prats$.

\begin{table}[!h]
\normalsize
$$
\begin{array}{ll}
  \mbox{(A1):} & \vdash_\e L_\e\phi\\
  \mbox{(A2):} & \vdash_\e L_{r+s}\phi\to L_r\phi\\
  \mbox{(A3):} & \vdash_\e L_r(\phi\land\psi)\land L_s(\phi\land\lnot\psi)\to L_{r+s-\e}\phi\\
  \mbox{(A4):} & \vdash_\e \lnot L_r(\phi\land\psi)\land \lnot L_s(\phi\land\lnot\psi)\to \lnot L_{r+s-\e}\phi\\
  \mbox{(R1):} & \mbox{If }\vdash_\e\phi\rightarrow\psi\mbox{ then }\vdash_\e L_r\phi\to L_r\psi\\
  \mbox{(R2):} & \{L_r\phi\mid\mbox{} r<s\}\vdash_\e L_s\phi\\
  \mbox{(R3):} & \{L_r\phi\mid\mbox{} r>s\}\vdash_\e \bot\\
\end{array}
$$\caption{\label{AS}The axiomatization of $\e$-provability for CML}
\end{table}

Axiom (A1) guarantees that the rate of any transition with an $\e$-approximation is at least $0+\e=\e$; this encodes the fact that the real measure of any set cannot be negative. (A2) states that if a rate is at least $r+s$ then it is at least $r$. (A3) and (A4) encode the additive properties of measures for disjoint sets: $\brck{\phi\land\psi}_\e$ and $\brck{\phi\land\lnot\psi}_\e$ are disjoint sets of processes such that $\brck{\phi\land\psi}_\e\cup\brck{\phi\land\lnot\psi}_\e=\brck\phi_\e$. The rule (R1) establishes the monotonicity of $L_r$. In this axiomatic system we have two infinitary rules, (R2) and (R3). The first reflects the Archimedian property of rationals: if it is possible a transition from a state to a given set of states at any rate $r<s$, then the rate of the transition is at least $s$. (R3) eliminates the possibility of having transitions at infinite rates.

Now we can complete the list of parametric meta-concepts initiated in Definition \ref{meta1}.

\begin{definition}\label{meta2}
A formula $\phi\in\lng$ is \emph{$\e$-provable}, written $\vdash_\e\phi$, if either it is an instance of an axiom or it can be proved from axioms using the proof rules. A formula $\phi\in\lng$ is \emph{$\e$-consistent}, if $\phi\to\bot$ is not provable. 

Given a set $\Phi\subseteq\lng$ of formulae, we say that $\Phi$ $\e$-proves $\phi$, denoted by $\Phi\vdash_\e\phi$, if $\phi$ can be proved from axioms and the formulae of $\Phi$.
$\Phi$ is \emph{$\e$-consistent} if it is not the case that $\Phi\vdash_\e\bot$. 

For a sublanguage $\lng'\subseteq\lng$, we say that $\Phi\in\lng$ is \emph{$\lng'$-maximally $\e$-consistent} if $\Phi$ is $\e$-consistent and no formula of $\lng'$ can be added to it without
making it $\e$-inconsistent.
\end{definition}

The next theorem states that $\models_\e$ and $\vdash_\e$ agree about the class of CMPs. As before, we will simply denote $\vdash_0$ by $\vdash$ and we call to it as \emph{classic provability}.

\begin{theorem}[Soundness and Weak Completeness]\label{completeness}
  The axiomatic system of $\e$-provability is sound and complete for
  the $\e$-semantics, i.e., for any $\phi\in\lng$, $$\vdash_\e\phi\mbox{ iff }\models_\e\phi.$$
\end{theorem}

\begin{proof}
  In \cite{Cardelli11a} we have shown that in table \ref{AS1} we have
  a sound and complete axiomatizaion of the classic provability for
  the classic semantics. In other words, we have proved that
  $\vdash\phi$ iff $\models\phi$.

\begin{table}[!h]
\normalsize
$$
   \begin{array}{ll}
        \mbox{(B1):} & \vdash L_0\phi\\
        \mbox{(B2):} & \vdash L_{r+s}\phi\to L_r\phi\\
        \mbox{(B3):} & \vdash L_r(\phi\land\psi)\land L_s(\phi\land\lnot\psi)\to L_{r+s}\phi\\
        \mbox{(B4):} & \vdash \lnot L_r(\phi\land\psi)\land \lnot L_s(\phi\land\lnot\psi)\to \lnot L_{r+s}\phi\\
        \mbox{(S1):} & \mbox{If }\vdash\phi\rightarrow\psi\mbox{ then }\vdash L_r\phi\to L_r\psi\\
        \mbox{(S2):} & \{L_r\phi \mid r<s\}\vdash L_s\phi\\
        \mbox{(S3):} & \{L_r\phi\mid r>s\}\vdash\bot\\
   \end{array}
$$\caption{\label{AS1}The axiomatic system of classic provability}
\end{table}

Obviously, the axioms (A1)-(A4) are the $\langle~\rangle_\e$-encodings
of the axioms (B1)-(B4) and similarly (R1)-(R3) are the encodings of
(S1)-(S3). 

Consequently, we obtain that $\vdash_\e\phi$ iff
$\vdash\langle\phi\rangle_\e$ and $\vdash\phi$ iff
$\vdash_\e\langle\phi\rangle^{\e}$. Now, in the light of Theorem
\ref{t2} we obtain $\vdash_\e\phi$ iff $\vdash\langle\phi\rangle_\e$
iff $\models\langle\phi\rangle_\e$ iff $\models_\e\phi$.
\end{proof}

Trivial consequences of the completeness theorem, Theorem \ref{t2} and
Lemma \ref{l1} are comprised in the next lemma which establishes the
relation between various $\e$-provabilities.

\begin{lemma}\label{l4}
For arbitrary $\phi\in\lng$, and rationals $\e,\e'\geq 0$,
\begin{enumerate}
\item $\vdash_{\e+\e'}\phi$ iff $\vdash_\e\langle\phi\rangle_{\e'}$;
  in particular, $\vdash_\e\phi$ iff $\vdash\langle\phi\rangle_{\e}$.
\item $\vdash_\e\phi$ iff $\vdash_{\e+\e'}\langle\phi\rangle^{\e'}$;
  in particular, $\vdash\phi$ iff
  $\vdash_{\e}\langle\phi\rangle^{\e}$.
\item If $\phi\in\lng^+$, $\vdash_{\e'}\phi$ implies
  $\vdash_{\e+\e'}\phi$.
\end{enumerate}
\end{lemma}

The previous lemma allows us to prove a parameterized deduction
theorem that establishes the frame in which one can use various
$\e$-provabilities relations in the same proof.

\begin{theorem}[Parameterized Deduction Theorem]\label{decidability}
For positive rationals $\e$ and $\e'$, and $\phi\in\lng^+$,
\begin{enumerate}
\item if $\vdash_{\e'+\e}(\phi\to\psi)$ and $\vdash_{\e'}\phi$, then
  $\vdash_{\e'+\e}\psi$;
\item if $\vdash_{\e'+\e}(\lnot\psi\to\lnot\phi)$ and
  $\vdash_{\e'}\phi$, then $\vdash_{\e'+\e}\psi$.
\end{enumerate}
\end{theorem}


The relation between the classic semantics and the $\e$-semantics also allows us to prove the next decidability result.

\begin{theorem}[Decidability and Complexity of $\e$-satisfiability]
The problem of deciding if an arbitrary property $\phi\in\lng$ is $\e$-satisfiable, i.e., if there exists a CMP $m$ such that $m\models_\e\phi$, is decidable in {\sc Pspace}. 
\end{theorem} 


Following the same proof line of Theorems \ref{completeness} and \ref{decidability} we can prove that the logic enjoys a parameterized version of finite model property.

\begin{theorem}[Finite model property]
Given an arbitrary $\e$-consistent formula $\phi\in\lng$, there exists a finite CMP $(\M,m)$ such that $m\models_\e\phi$.
\end{theorem}




\section{Behavioral Properties}\label{sec:behavior}

In the previous sections we developed the parametric metatheory and prove that it enjoys most of the metaproperties of the classic metatheory. In this section we investigate the relationship between this parametric logical framework and the behavioral properties of CMPs. We begin by recalling a result proved in \cite{Cardelli11a}.

\begin{theorem}[Logical characterization of bisimulation]\label{modalcharact}
  Let $\M=(M,\Sigma,\tau)$ be a CMK and $m,m'\in M$. The following
  assertions are equivalent.
  \begin{enumerate}
  \item $m\sim m'$; 
  \item For any $\phi\in\lng$, $m\models\phi$ iff $m'\models\phi$;
  \item For any $\phi\in\lng^+$, $m\models\phi$ iff $m'\models\phi$.
  \end{enumerate}
\end{theorem}

Because the encodings $\langle~\rangle_\e$ and $\langle~\rangle^\e$ preserve the logical implication, a consequence of the fact that CML characterizes stochastic bisimulation is the next theorem.

\begin{theorem}[Parameterized characterization of bisimulation]\label{paramcharact}
For arbitrary $\e\geq 0$, the following assertions are equivalent. 
\begin{enumerate}
\item $m\sim m'$;
\item For any $\phi\in\lng$, $m\models_\e\phi$ iff $m'\models_\e\phi$;
\item For any $\phi\in\lng^+$, $m\models_\e\phi$ iff $m'\models_\e\phi$.
\end{enumerate}
\end{theorem}

However, the interrelations between various $\e$-semantics allow us to
prove some stronger results. For this, in what follows, we extend the
concept of bisimulation towards a notion of $\e$-orders that reflect
the approximated behaviors. Similar notions have been proposed in
literature for less general models. It is the case of the point-wise
simulation distance defined in \cite{Jou90} for discrete
probabilistic systems and similarly in \cite{Alfaro03,Alfaro09} for
weighted systems, and for general probabilistic systems in
\cite{vanBreugel01b,Desharnais04}.

Recall that for a set $C$ and a relation $R$, $C^R$ denotes the
closure of $C$ to $R$.

\begin{definition}[$\epsilon$-behavioral orders]
  Given a CMK $\M=(M,\Sigma,\theta)$, a relation $R\subseteq M\times
  M$ closed under bisimulation is
  \begin{itemize}
  \item an \emph{$\e$-behavioral order} whenever $m \mathrel{R} n$,
    implies that for any $C\in G_\M$, $$\theta(n)(C)-\theta(m)(
    C^R)\leq\e.$$
  \item an \emph{essential $\e$-behavioral order} whenever $m
    \mathrel{R} n$, implies that for any $C\in
    \ol{G_\M}$, $$\theta(n)(C)-\theta(m)( C^R)\in[0,\e].$$
  \end{itemize}
We use $\prec_e$ to denote the largest $\e$-behavioral order and
$\prec^+_\e$ to denote the largest essential $\e$-behavioral order.
\end{definition}

Observe that, in the definition of $\e$-behavioural order we can have
$\theta(n)(C)<\theta(m)(C^R)$, while for essential $\e$-behavioural
order we have always that $\theta(n)(C)\geq\theta(m)( C^R)$. Notice
also that both $\prec_\e$ and $\prec^+_\e$ are not equivalences and
that an essential $\e$-behavioral order is an $\e$-behavioral order,
i.e., $\mathord{\prec^+_\e}\subseteq \mathord{\prec_\e}$.

\begin{example}
  Figure~\ref{fig:example} shows three discrete processes with initial
  states $m$,$n$ and $o$ respectively. Their mutual relationship is
  easily shown by producing an $\e$-order. For simplicity, we have not
  represented the transitions with rate
  $0$. 
  Assuming that $s + s' = t$, we
  obtain $$R=\{(m,n),(m_1,n_1),(m_2,n_2),(m_4,n_2),(m_3,n_3),(m_5,n_3)\}\subseteq\prec_{2\e},$$
  i.e., $m\prec_{2\e} n$ and the value $2\e$ is obtained
  from $$\theta(n)(C)-\theta(m)(C^R)=2\e$$ for
  $C=\{m_1,m_2,m_4,n_1,n_2\}$. Observe in this case that
  $m_4\prec_{2\e} n_2$ even if the rate of exiting $n_2$ is smaller
  than the rate of exiting $m_4$. But for this reason we do not have
  $m\prec^+_{\e'}n$ for all $\e'>0$.
  
  Similarly, we obtain $m \prec^+_\e o$
  for $$R=\{(m,o),(m_i,o_i)_{i=1..5}\}\subseteq\prec^+_\e.$$

  In other words, the rates of $m$ are at most $2\e$-smaller than the
  corresponding rates of $n$, or larger; and $m$ has at most
  $\e$-smaller rates than the corresponding ones of $o$, but not
  larger.
  \end{example}

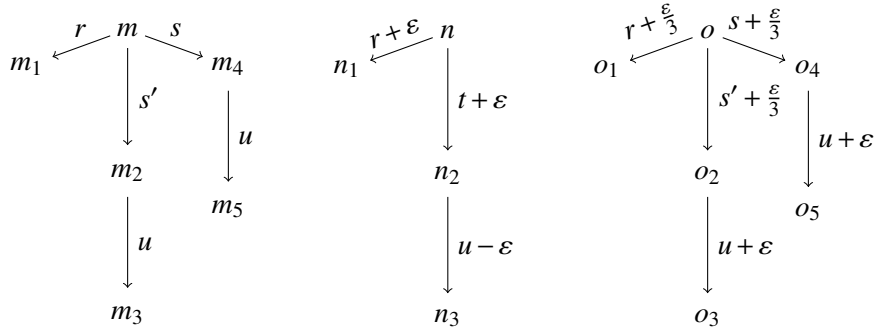
\begin{figure}[t]
  \centering
      \begin{tikzpicture}[->,shorten >=1pt, auto, node distance=.75cm,
        initial text=, scale=0.75]
        
        \begin{scope}[shape=circle, outer sep=1pt,minimum
          size=5pt,inner sep=1pt, node distance=1.9cm] 
          \node (t0) {$m$};        
          \node (t1) at (200:1.9){$m_1$};
          \node (t5) [below of=t0] {$m_2$};
          \node (t4) [below of=t5] {$m_3$};
          \node (t6) at (340:1.9){$m_4$};
          \node (t7) [below of=t6]{$m_5$};
        \end{scope}        
	            
        \begin{scope}
          \draw (t0) -- node[above] {$r$} (t1);
          \draw (t0) -- node[above] {$s$} (t6);
          \draw (t0) -- node[right] {$s'$} (t5);
          \draw (t5) -- node[right] {$u$} (t4);
          \draw (t6) -- node[right] {$u$} (t7);
       \end{scope}

     \end{tikzpicture}
\qquad%
      \begin{tikzpicture}[->,shorten >=1pt, auto, node distance=.75cm,
        initial text=, scale=0.75]
        
        \begin{scope}[shape=circle, outer sep=1pt,minimum
          size=5pt,inner sep=1pt, node distance=1.9cm] 
          \node (t0) {$n$};        
          \node (t1) at (200:1.9){$n_1$};
          \node (t5) [below of=t0] {$n_2$};
          \node (t4) [below of=t5] {$n_3$};
        \end{scope}

        \begin{scope}
          \draw (t0) -- node[above,sloped] {$r+\e$} (t1);
          \draw (t0) -- node[right] {$t+\e$} (t5);
          \draw (t5) -- node[right] {$u-\e$} (t4);
       \end{scope}
     \end{tikzpicture}
\qquad%
      \begin{tikzpicture}[->,shorten >=1pt, auto, node distance=.75cm,
        initial text=, scale=0.75]
        
        \begin{scope}[shape=circle, outer sep=1pt,minimum
          size=5pt,inner sep=1pt, node distance=1.9cm] 
          \node (t0) {$o$};        
          \node (t1) at (200:1.9){$o_1$};
          \node (t5) [below of=t0] {$o_2$};
          \node (t4) [below of=t5] {$o_3$};
          \node (t6) at (340:1.9){$o_4$};
          \node (t7) [below of=t6]{$o_5$};
        \end{scope}

        \begin{scope}
          \draw (t0) -- node[above,sloped] {$r+\frac{\e}{3}$} (t1);
          \draw (t0) -- node[right] {$s'+\frac{\e}{3}$} (t5);
          \draw (t5) -- node[right] {$u+\e$} (t4);
          \draw (t0) -- node[above] {$s+\frac{\e}{3}$} (t6);
          \draw (t6) -- node[right] {$u+\e$} (t7);
       \end{scope}
     \end{tikzpicture}
     \caption{Systems demonstrating $\e$-behavioral orders and the
       significance of stochastic transitions.}
  \label{fig:example}
\end{figure}


Applying Theorem \ref{t1} on bisimulation generators, we obatin the
following result, which shows the concept of $\e$-behavioral order
generalizes the concept of stochastic bisimulation.

\begin{lemma}\label{l5}
Any bisimulation relation is an $\e$-behavioral order for any rational $\e\geq 0$. Moreover, if $m\sim n$ then $m\prec_0 n$ and $n\prec_0 m$.
\end{lemma}

The next theorem generalizes the theorems \ref{modalcharact} and \ref{paramcharact} for behavioral orders. In this new context the $\e$-semantics is the key.

\begin{theorem}[Logical characterization of $\prec_\e$]\label{characterization}
  For arbitrary rational $\e\geq0$, $$[\mbox{for any }\phi\in\lng^+,
  n\models\phi\mbox{ implies }m\models_\e\phi]\mbox{ iff }m\prec_\e
  n.$$
\end{theorem}

The previous theorem can be further generalized to comprise also the
negative formulae. In order to do that, because there is an asymmetry
between the behavior of the positive and negative formulae in the
$\e$-semantics, we will need an extra encoding that we define
below. This encoding assumes that formulae are in disjunctive
normal form (when $L_r\phi$ are considered atoms).

$$\begin{array}{ll}
|\phi|_\e= & \left\{
\begin{array}{ll}
\top & \textrm{if } \phi=\top\\
\bot & \textrm{if } \phi=\bot\\
L_r|\psi|_\e & \textrm{if } \phi=L_r\psi\\
\lnot L_{r+\e}|\psi|_\e & \textrm{if }\phi=\lnot L_r\psi\\
|\psi_1|_\e\land|\psi_2|_\e & \textrm{if } \phi=\psi_1\land\psi_2\\
|\psi_1|_\e\lor|\psi_2|_\e & \textrm{if } \phi=\psi_1\lor\psi_2

\end{array}\right. 
\end{array}$$ 

With this encoding we can state the generalization of the previous theorem. 

\begin{theorem}[Logical characterization of $\prec^+_\e$]\label{generalization}
For arbitrary rational $\e\geq0$, $$[\mbox{for any }\phi\in\lng, n\models\phi\mbox{ implies }m\models_\e|\phi|_\e]\mbox{ iff }m\prec^+_\e n.$$
\end{theorem}

\section{The pseudometrizable space of processes}

In what follows we use $\prec_\e$ to define a canonical distance between
CMPs that resemble (and generalize) the well known point-wise
distances for particular types of CMPs such as Markov chains.

The next result guarantees that between any two systems there is a
$\prec_\e$ relation for an $\e$ that is big enough.

\begin{lemma}
For any pair of CMPs $m$ and $n$ there exists a positive rational $\e$ such that $m\prec_\e n$.
\end{lemma}




This lemma allows us to define a function $$d:M\times M\to\preals\mbox{ by }d(m,n)=inf\{\e\mid m\prec_\e n\mbox{ and }n\prec_\e m\}.$$ 
As stated in the next theorem, $d$ is a pseudometric on $M$ that measures how different two systems are from the point of view of their behavior. The distance between two systems is 0 iff the systems are bisimilar.

\begin{theorem}[Pseudometric]\label{pseudometric}
  The function $d:M\times M\to\preals$ defined before is a
  pseudometric on $M$ which characterizes stochastic bisimulation,
  i.e., $$d(m,n)=0\mbox{ iff }m\sim n.$$
\end{theorem}

To conclude this section and understand the significance of this distance, we shall take a look at the following example.

\begin{example}
Consider the CMPs described in Figure \ref{ex2}. We notice that the processes with initial states $m$ and $o$ are quite similar with respect to structure and rate values; the second one has all the transitions $\e$-bigger than the first one. So a first guess will be that $d(m,o)=\e$. But this is not the case because the rate of exiting the state $m$ is $$\theta(m)(\{m_1,m_2,m_3,m_4,m_5\})=r+s+s',$$ which is $3\e$ smaller than the rate of exiting the state $o$ $$\theta(o)(\{o_1,o_2,o_3,o_4,o_5\})=r+s+s'+3\e.$$ Consequently, $d(m,o)=3\e$.

Consider now the CMPs with $m$ and $n$ as initial states and suppose, as before, that $s+s'=t$. We should notice this time that not all the transitions of the first CMP are bigger than the transitions of the second. However, every pair of transitions do not differ with more than $\e$. Since there are paired transitions that differ with exactly $\e$ value, we obtain that $d(m,n)=\e$. 
\end{example}

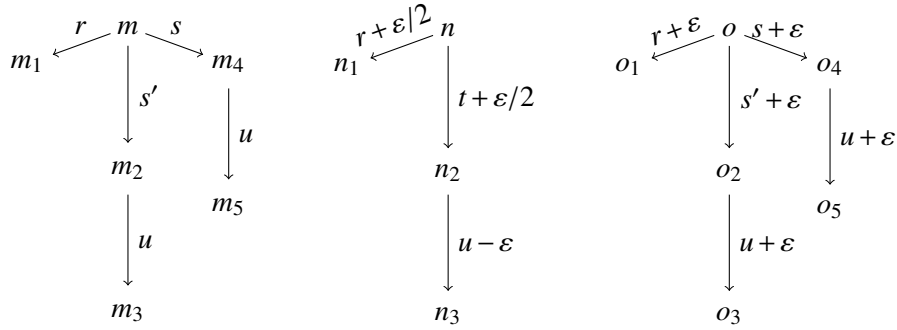
\begin{figure}[t]
  \centering
      \begin{tikzpicture}[->,shorten >=1pt, auto, node distance=.75cm,
        initial text=, scale=0.75]
        
        \begin{scope}[shape=circle, outer sep=1pt,minimum
          size=5pt,inner sep=1pt, node distance=1.9cm] 
          \node (t0) {$m$};        
          \node (t1) at (200:1.9){$m_1$};
          \node (t5) [below of=t0] {$m_2$};
          \node (t4) [below of=t5] {$m_3$};
          \node (t6) at (340:1.9){$m_4$};
          \node (t7) [below of=t6]{$m_5$};
        \end{scope}        
	            
        \begin{scope}
          \draw (t0) -- node[above] {$r$} (t1);
          \draw (t0) -- node[above] {$s$} (t6);
          \draw (t0) -- node[right] {$s'$} (t5);
          \draw (t5) -- node[right] {$u$} (t4);
          \draw (t6) -- node[right] {$u$} (t7);
       \end{scope}

     \end{tikzpicture}
\qquad%
      \begin{tikzpicture}[->,shorten >=1pt, auto, node distance=.75cm,
        initial text=, scale=0.75]
        
        \begin{scope}[shape=circle, outer sep=1pt,minimum
          size=5pt,inner sep=1pt, node distance=1.9cm] 
          \node (t0) {$n$};        
          \node (t1) at (200:1.9){$n_1$};
          \node (t5) [below of=t0] {$n_2$};
          \node (t4) [below of=t5] {$n_3$};
        \end{scope}

        \begin{scope}
          \draw (t0) -- node[above,sloped] {$r+\e/2$} (t1);
          \draw (t0) -- node[right] {$t+\e/2$} (t5);
          \draw (t5) -- node[right] {$u-\e$} (t4);
       \end{scope}
     \end{tikzpicture}
\qquad%
      \begin{tikzpicture}[->,shorten >=1pt, auto, node distance=.75cm,
        initial text=, scale=0.75]
        
        \begin{scope}[shape=circle, outer sep=1pt,minimum
          size=5pt,inner sep=1pt, node distance=1.9cm] 
          \node (t0) {$o$};        
          \node (t1) at (200:1.9){$o_1$};
          \node (t5) [below of=t0] {$o_2$};
          \node (t4) [below of=t5] {$o_3$};
          \node (t6) at (340:1.9){$o_4$};
          \node (t7) [below of=t6]{$o_5$};
        \end{scope}

        \begin{scope}
          \draw (t0) -- node[above,sloped] {$r+\e$} (t1);
          \draw (t0) -- node[right] {$s'+\e$} (t5);
          \draw (t5) -- node[right] {$u+\e$} (t4);
          \draw (t0) -- node[above] {$s+\e$} (t6);
          \draw (t6) -- node[right] {$u+\e$} (t7);
       \end{scope}
     \end{tikzpicture}
     \caption{Distances between Markovian processes.}
  \label{ex2}
\end{figure}


\section{Conclusions}

In this paper we have introduced a parametric metatheory for Continuous Markovian Logic. The parameter $\e$ of the metatheory encodes an observation error that might appear when we analyze a stochastic system. We define an $\e$-semantics and an axiomatized $\e$-proof system and we show that the $\e$-provability relation is sound and complete with respect to the $\e$-satisfiability relation. This entire logical framework also allows us to transfer metaproperties between various $\e$-levels of the metatheory. We prove a series of results regarding the connection between $\e$-satisfiability and $\e+\e'$-satisfiability and a parameterized deduction theorem that combines $\e$-provability and $\e+\e'$-provability results.

This classic metalogical framework allows us to give an uniform treatment to all logical properties, including the ones involving negative or logical implication, while avoiding unorthodox logical constructs as the real-valued
logics. The framework also supports us in identifying two canonical behavioural orders that extend stochastic bisimulation and organize the space of CMPs. These bisimulation orders are the cornerstones in the definition of a pseudometric on CMPs that measure the behavioral similarity of processes.

The metalogical framework introduced in this paper can be particularized to more specific Markovian models such as the discrete or continuous-time Markov chains. Moreover, the entire development can be adapted to specialize on the probabilistic cases, as the mathematical structure that supports the definition of CMPs is similar to the one that supports the definition of labelled Markov processes in the form of \cite{Panangaden09}.

This paper opens a series of interesting research questions regarding the relationship between $\e$-satisfiability, $\e$-provability and metric semantics. There are many open questions related to the possibility of defining a pseudometric over the class of logical formulae that shall measure $\e$-provability; for instance such that the distance between $\phi$ and $\psi$ is $0$ iff $\phi$ and $\psi$ are logical equivalent. On this direction we expect to be able to prove a version of metric completeness that relates the pseudometric space of CMPs to the pseudometric space of logical formulae. The first two authors in collaboration with Prakash Panangaden have already obtained a series of results in this direction \cite{Larsen12a,Larsen12b}; but these results do not involve the parametric metatheory yet. The hope is that the new metatheoretical perspective introduced in this paper will eventually solve some of the open problems that resisted to the other approaches.


\section*{Acknowledgement}
This research was supported by the VKR Center of Excellence MT-LAB and by the Sino-Danish Basic Research Center IDEA4CPS. Mardare was also supported by Sapere Aude: DFF-Young Researchers Grant 10-085054 of the Danish Council for Independent Research. 

Mardare would like to thank Prakash Panangaden for discussions about various aspects of the theory of Markov processes and logics that eventually allowed us to arrive to the current level of understanding of these problems. He is also grateful to Luca Cardelli, Gordon Plotkin and Vincent Danos for discussions in the past about various aspects of Markovian logics that led his research to these results. 




\begin{thebibliography}{vBMOW03}

\bibitem[AFS09]{Alfaro09} L.~de~Alfaro, M.~Faella, and M.~Stoelinga.
  \newblock Linear and branching system metrics.  \newblock \emph{IEEE
    Trans. Software Eng.}, vol. 35(2) pp. 258--273,
  2009. \doi{10.1109/TSE.2008.106}

\bibitem[AHM03]{Alfaro03} L.~de~Alfaro, T.~A. Henzinger, and
  R.~Majumdar.  \newblock Discounting the future in systems theory.
  \newblock in \emph{ICALP03}, pp. 1022--1037,
  2003. \doi{10.1007/3-540-45061-0\_79}

\bibitem[Aum99]{Aumann99a} R. Aumann.  \newblock Interactive
  epistemology {I}: knowledge.  \newblock {\em International Journal
    of Game Theory}, vol. 28 pp. 263--300,
  1999. \doi{10.1007/s001820050111}

\bibitem[BMM09]{Ballarini09} P. Ballarini, R. Mardare, I. Mura.
  \newblock Analysing Biochemical Oscillations through Probabilistic
  Model Checking.  \newblock In {\em FBTC 2008}, ENTCS vol. 229(1)
  pp. 3--19, 2009. \doi{10.1016/j.entcs.2009.02.002}

\bibitem[CLM11a]{Cardelli11a} L. Cardelli, K.~G. Larsen, and
  R. Mardare.  \newblock Continuous Markovian logic - from complete
  axiomatization to the metric space of formulas.  \newblock In {\em
    CSL}, pp. 144--158, 2011. \doi{10.4230/LIPIcs.CSL.2011.144}

\bibitem[CLM11b]{Cardelli11b} L. Cardelli, K.~G. Larsen, and
  R. Mardare.  \newblock Modular Markovian logic.  \newblock In {\em
    ICALP (2)}, pp. 380--391,
  2011. \doi{10.1007/978-3-642-22012-8\_30}

\bibitem[DEP02]{Desharnais02} J.~Desharnais, A.~Edalat, and
  P.~Panangaden.  \newblock Bisimulation for labelled {Markov}
  processes.  \newblock {\em I$\&$C}, Vol. 179(2) pp. 163--193,
  2002. \doi{10.1006/inco.2001.2962}

\bibitem[D+04]{Desharnais04} J. Desharnais, V. Gupta, R. Jagadeesan,
  and P. Panangaden.  \newblock A metric for labelled {Markov}
  processes.  \newblock {\em TCS}, vol. 318(3) pp. 323--354, June
  2004. \doi{10.1016/j.tcs.2003.09.013}

\bibitem[DLT08]{Desharnais08} J.~Desharnais, F.~Laviolette, and
  M.~Tracol.  \newblock Approximate analysis of probabilistic
  processes: Logic, simulation and games.  \newblock \emph{QEST08},
  IEEE Computer Society, pp. 264--273,
  2008. \doi{10.1109/QEST.2008.42}

\bibitem[DP03]{Desharnais03b} J. Desharnais and P. Panangaden.
  \newblock Continuous stochastic logic characterizes bisimulation for
  continuous-time {M}arkov processes.  \newblock {\em JLAP}, vol. 56
  pp. 99--115, 2003. \doi{10.1016/S1567-8326(02)00068-1}

\bibitem[Dob07]{Doberkat07} E.-E. Doberkat.  \newblock
  \emph{Stochastic Relations. Foundations for Markov Transition
    Systems.}  \newblock Chapman and Hall, New York, 2007. 

\bibitem[dVR99]{deVink99} E.~de~Vink and J.~J. M.~M. Rutten.
  \newblock Bisimulation for probabilistic transition systems: A
  coalgebraic approach.  \newblock {\em TCS}, vol. 221(1/2)
  pp. 271--293, 1999. \doi{10.1016/S0304-3975(99)00035-3}

\bibitem[FH94]{Fagin94} R.~Fagin and J.~Y. Halpern.  \newblock
  Reasoning about knowledge and probability.  \newblock {\em JACM},
  vol.~41(2) pp. 340--367, 1994. \doi{10.1145/174652.174658}

\bibitem[FGK10]{FGK10} D.~Fischer, E.~Gr{\"a}del, and L.~Kaiser.
  \newblock Model checking games for the quantitative
  {$\mu$}-calculus.  \newblock \emph{Theory Comput. Syst.},
  vol.~47(3), pp. 696--719, 2010. \doi{10.1007/s00224-009-9201-y}

\bibitem[FLT10]{FLT10} U.~Fahrenberg, K.~G. Larsen and C.~Thrane.
  \newblock A Quantitative Characterization of Weighted Kripke
  Structures in Temporal Logic.  \newblock \emph{Computing and
    Informatics}, vol. 29(6+) pp. 1311--1324, 2010. 

\bibitem[MHP05]{Henzinger05} T.~A. Henzinger, R.~Majumdar, and
  V.~S. Prabhu, \newblock Quantifying similarities between timed
  systems, \newblock in \emph{FORMATS05}, pp. 226--241,
  2005. \doi{10.1007/11603009\_18}

\bibitem[JS90]{Jou90} C.-C. Jou and S.~A. Smolka.  \newblock
  Equivalences, congruences, and complete axiomatizations for
  probabilistic processes.  \newblock In {\em CONCUR},
  1990. \doi{10.1007/BFb0039071}

\bibitem[Koz85]{mu} D.~Kozen.  \newblock A probabilistic {PDL}.
  \newblock Journal of Computer and Systems Sciences, vol. 30(2):
  pp. 162--178, 1985.\doi{10.1016/0022-0000(85)90012-1}


\bibitem[KP10]{KP} C. Kupke, D. Pattinson.  \newblock On Modal Logics
  of Linear Inequalities.  \newblock In Proceedings of AiML 2010. 


\bibitem[LS91]{Larsen91} K.~G. Larsen and A.~Skou.  \newblock
  Bisimulation through probablistic testing.  \newblock {\em
    Information and Computation}, vol. 94, pp. 1--28,
  1991. \doi{10.1016/0890-5401(91)90030-6}

\bibitem[LMP12a]{Larsen12a} K. G. Larsen, R. Mardare, P. Panangaden
  \newblock A metric analogue of Stone duality for Markov processes
  \newblock unpublished manuscript available from
  http://people.cs.aau.dk/$\sim$mardare

\bibitem[LMP12b]{Larsen12b} K. G. Larsen, R. Mardare, P. Panangaden
  \newblock Taking it to the limit: Approximate reasoning for Markov
  processes \newblock unpublished manuscript available from
  http://people.cs.aau.dk/$\sim$mardare


\bibitem[MV04]{Moss04} L.~S. Moss and I.~D. Viglizzo.  \newblock
  {H}arsanyi type spaces and final coalgebras constructed from
  satisfied theories.  \newblock {\em ENTCS}, vol. 106, pp. 279--295,
  2004. \doi{10.1016/j.entcs.2004.02.036}

\bibitem[Pan09]{Panangaden09} P. Panangaden.  \newblock {\em Labelled
    {M}arkov Processes}.  \newblock Imperial College Press, 2009.

\bibitem[TFL10]{TFL10} C.~Thrane, U.~Fahrenberg and K.~G. Larsen,
  \newblock Quantitative Simulations of Weighted Transition Systems
  \newblock \emph{Journal of Logic and Algebraic Programming},
  vol. 79(7), pp. 689--703, 2010. \doi{10.1016/j.jlap.2010.07.010}

\bibitem[Thr11]{crt}
C.~Thrane.
\newblock Quantitative Models and Analysis for Reactive Systems.
\newblock PhD Thesis, Aalborg University, 2011.

\bibitem[vB+03]{vanBreugel03}
F. van Breugel, M. Mislove, J. Ouaknine, and J. Worrell.
\newblock An intrinsic characterization of approximate probabilistic
  bisimilarity.
\newblock In {\em FOSSACS 03}, 2003.

\bibitem[vBW01]{vanBreugel01b}
F. van Breugel and J. Worrell.
\newblock An algorithm for quantitative verification of probabilistic systems.
\newblock In {\em CONCUR'01}, pp. 336--350, 2001.

\bibitem[Zho07]{Zhou07} C.~Zhou.  \newblock {\em A complete deductive
    system for probability logic with application to {H}arsanyi type
    spaces}.  \newblock PhD thesis, Indiana University, 2007.

\end{thebibliography}
\end{document}